\documentclass[11pt]{article}
\usepackage{amsmath,amssymb, amsthm,amsfonts}
\usepackage{microtype}
\usepackage{color}
\usepackage{mathptmx}
\usepackage{fullpage}
\usepackage{mathrsfs}  
\usepackage{natbib}
\usepackage{dsfont}

\newtheorem{theorem}{Theorem}

\newtheorem{proposition}[theorem]{Proposition}

\newcommand{\cclass}[1]{\textsf{#1}}

\newcommand{\set}[1]{\left\{ #1 \right\}}

\newcommand{\e}{\mathrm{e}}

\newcommand{\aut}{\mathrm{Aut}} 

\newcommand{\GL}{\mathsf{GL}} 

\newcommand{\FF}{\mathbb{F}}

\newcommand{\RM}{\mathsf{RM}} 

\DeclareMathAlphabet{\varmathbb}{U}{bbold}{m}{n}

\newcommand{\one}{\mathds{1}}

\title{Quantum Fourier sampling, Code Equivalence, and the quantum security of the McEliece and Sidelnikov cryptosystems}
\author{Hang Dinh\\ 
Indiana University South Bend\\
\textsf{hdinh@cs.iusb.edu} 
\and
Cristopher Moore\\
University of New Mexico\\
and Santa Fe Institute\\
\textsf{moore@cs.unm.edu} 
\and  
Alexander Russell\\
University of Connecticut\\
\textsf{acr@cse.uconn.edu}
}

\begin{document}

\maketitle              

\begin{abstract}
  The Code Equivalence problem is that of determining whether two
  given linear codes are equivalent to each other up to a permutation of the coordinates. This problem has a direct reduction to a
  nonabelian hidden subgroup problem (HSP), suggesting a possible quantum algorithm 
  analogous to Shor's algorithms for factoring or discrete log.  However, in~\citet{Ref_Dinh11McEliece} we showed that in many cases of interest---including Goppa codes---solving this case of the HSP requires rich, entangled measurements.
  Thus, solving these cases of Code Equivalence via Fourier sampling appears to be out of reach of current families of quantum algorithms.

  Code equivalence is directly related to the security of McEliece-type cryptosystems in the case where the private code is known to the adversary.  However, for many codes the support splitting algorithm of Sendrier provides a classical attack in this case.  We revisit the claims of our previous article in the light of these classical attacks, and discuss the particular case of the Sidelnikov cryptosystem, which is based on Reed-Muller codes.  \end{abstract}

\pagenumbering{arabic}
\section{Introduction}

Code Equivalence is the problem of deciding whether two matrices over a finite field generate equivalent linear codes, i.e., codes that are equal up to a fixed permutation on the codeword coordinates. \citet{Ref_Petrank97code} showed that Code Equivalence is unlikely to be \cclass{NP}-complete, but is at least as hard as Graph Isomorphism. We consider a search version of Code Equivalence: given generator matrices $M$ and $M'$ for two equivalent linear $q$-ary codes, find a pair of matrices $(S, P)$, where $S$ is an invertible square matrix over $\FF_q$ and $P$ is a permutation matrix, such that $M'=SMP$.   

Code Equivalence has an immediate presentation as a \emph{hidden subgroup problem}, suggesting that one might be able to develop an efficient quantum algorithm for it via the quantum Fourier transform. In our previous article~\citep{Ref_Dinh11McEliece}, however, we showed that under natural structural assumptions on the code, the resulting instance of the hidden subgroup problem requires entangled measurements of the coset state and, hence, appears to be beyond the reach of current methods.

We argued in~\citep{Ref_Dinh11McEliece} that our results strengthen the case for the McEliece cryptosystem as a candidate for 
\emph{post-quantum cryptography}---a cryptosystem that can be implemented with classical computers, 
but which will remain secure even if and when quantum computers are built.  In this note, we revisit this statement in light of Sendrier's support splitting algorithm (SSA), which finds the hidden permutation $P$ for many families of codes.  In particular, the SSA implies that the McEliece cryptosystem based on Goppa codes is classically insecure when the private code is known.  We also observe that our results apply to Reed-Muller codes and thus to a natural quantum attack on the Sidelnikov cryptosystem.

\section{Ramifications for McEliece-type cryptosystems}

The private key of a McEliece cryptosystem is a triple $(S, M, P)$, where $S$ is an invertible matrix over $\FF_q$, $P$ is a permutation matrix, and $M$ is the generator matrix for a $q$-ary error-correcting code that permits efficient decoding. The public key is the generator matrix $M' = SMP$. If both $M$ and $M'$ are known to an adversary, the problem of recovering $S$ and $P$ (the remainder of the secret key) is precisely the version of Code Equivalence described above.  If $M$ and $M'$ have full rank, then given $P$ we can find $S$ by linear algebra.  Thus the potentially hard part of the problem is finding the hidden permutation $P$.

We call an adversary apprised of both $M$ and $M'$ a \emph{known-code} adversary.  In our recent article~\citep{Ref_Dinh11McEliece}, we noted that our results on Goppa codes imply that the natural quantum attack available to a known-code adversary yields hard cases of the hidden subgroup problem, and asserted that this should bolster our confidence in the post-quantum security of the McEliece cryptosystem.  

However, the classical \emph{support splitting algorithm} (SSA) of~\citet{Ref_Sendrier00finding} can efficiently solve Code Equivalence for Goppa codes, and many other families of codes as well.  
(In addition, Goppa codes of high rate can be distinguished from random codes, opening them to additional attacks~\citep{FOPTDist_YACC10}.)
Thus for McEliece based on Goppa codes, the known-code adversary is too powerful: it can break the cryptosystem classically. Therefore, the hardness of the corresponding instances of the HSP has little bearing on the post-quantum security of the McEliece cryptosystem, at least for this family of codes.

The situation is similar in many ways to the status of Graph Isomorphism.  There is a natural reduction from Graph Isomorphism to the HSP on the symmetric group, but a long series of results (e.g., \citet{Ref_HallgrenJACM,sieve}) have shown that the resulting instances of the HSP require highly-entangled measurements, and that known families of such measurements cannot succeed.  Thus the miracle of Shor's algorithms for factoring and discrete log, where we can solve these problems simply by looking at the symmetries of a certain function, does not seem to apply to Graph Isomorphism.  Any efficient quantum algorithm for it would have to involve significantly new ideas.

On the other hand, many cases of Graph Isomorphism are easy classically, including graphs with bounded eigenvalue multiplicity~\citep{bounded-mult} and constant degree~\citep{luks}.  Many of these classical algorithms work by finding a \emph{canonical labeling} of the graph~\citep{babai-canonical,babai-luks}, giving each vertex a unique label based on local quantities.  These labeling schemes use the details of the graph, and not just its symmetries---precisely what the reduction to the HSP leaves out.  Analogously, the support splitting algorithm labels each coordinate of the code by the weight enumerator of the hull of the code punctured at that coordinate.  For most codes, including Goppa codes, this creates a labeling that is unique or nearly unique, allowing us to determine the permutation $P$.  

There are families of instances of Graph Isomorphism that defeat known methods, due to the fact that no local or spectral property appears to distinguish the vertices from each other.  In particular, no polynomial-time algorithm is known for isomorphism of strongly regular graphs.  (On the other hand, these graphs are highly structured, yielding canonical-labeling algorithms that, while still exponential, are faster than those known for general graphs~\citep{spielman-stronglyreg}.)  In the same vein, we might hope that there are families of codes where the coordinates are hard to distinguish from each other using linear-algebraic properties.  In that case, the corresponding McEliece cryptosystem might be hard classically even for known-code adversaries, and the reduction to the HSP would be relevant to their post-quantum security.

Along these lines, \citet{Ref_Sidelnikov94public} proposed a variant of the McEliece cryptosystem using binary Reed-Muller codes.  Since there is a single Reed-Muller code of given rate and block length, the code $M$ is known to the adversary and the security of the system is directly related to the Code Equivalence problem. Additionally, since Reed-Muller codes are self-dual, they coincide with their hulls so that the weight enumerators used by the SSA are exponentially large, making them resistant to that classical attack.  There are also extremely efficient algorithms for error correction in Reed-Muller codes, suggesting that large key sizes are computationally feasible.

We observe below that the results of~\cite{Ref_Dinh11McEliece} apply directly to Reed-Muller codes, and thus frustrate the natural quantum Fourier sampling approach to the corresponding instances of Code Equivalence. As virtually all known exponential speed-ups of quantum algorithms for algebraic problems derive from Fourier sampling, this suggests that new ideas would be necessary to exploit quantum computing for breaking the Sidelnikov system.

On the other hand, a classical algorithm of~\citet{Ref_Minder07Cryptanalysis} solves the Code Equivalence problem for binary Reed-Muller codes in quasipolynomial time, at least in the low-rate setting where Reed-Muller codes have the best performance, yielding a direct attack on the Sidelnikov system.

\section{Quantum Fourier sampling for Code Equivalence and Reed-Muller codes}

We say a linear code $M$ is \emph{HSP-hard} if strong quantum Fourier sampling, or more generally any measurement of a coset state, reveals negligible information about the permutation between $M$ and any code equivalent to $M$.  (See e.g.~\citet{mrs-symmetric} for definitions of the coset state and strong Fourier sampling.)  If $M$ is a $q$-ary $[n,k]$-code, its automorphism group $\aut(M)$ is the set of permutations $P \in S_n$ such that $M=SMP$ for some invertible $k \times k$ matrix $S$ over $\FF_q$.  Recall that the \emph{support} of a permutation $\pi\in S_n$ is the number of points that are not fixed by $\pi$, and the \emph{minimal degree} of a subgroup $H \subseteq S_n$ is the smallest support of any non-identity $\pi \in H$.  The following theorem is immediate from Corollary 1 in~\citet{Ref_Dinh11McEliece}:
\begin{theorem}\label{Thm:Quantum-McEliece}
Let $M$ be a $q$-ary $[n,k]$-linear code such that $q^{k^2}\leq n^{0.2 \,n}$. If $|\aut(M)| \le \e^{o(n)}$ and the minimal degree of $\aut(M)$ is $\Omega(n)$, then $M$ is HSP-hard.
\end{theorem}

We apply Theorem~\ref{Thm:Quantum-McEliece} to binary Reed-Muller codes as follows.  Let $m$ and  $r$ be positive integers with $r<m$, and let $n=2^m$. Fix an ordered list $(\alpha_1,\ldots, \alpha_n)$ of all $2^m$ binary vectors of length $m$, i.e., $\FF_2^m=\set{\alpha_1,\ldots, \alpha_n}$. The $r^{\rm th}$-order binary Reed-Muller code of length $n$, denoted $\RM(r, m)$, consists of  codewords of the form $(f(\alpha_1),\ldots, f(\alpha_n))$, where $f \in \FF_2[X_1,\ldots, X_m]$ ranges over all binary polynomials on $m$ variables of degree at most $r$. The code $\RM(r,m)$ has dimension equal to the number of monomials of degree at most $r$, 
\[
k=\sum_{j=0}^r {m\choose j}\,.
\]
To apply Theorem~\ref{Thm:Quantum-McEliece}, we first need to choose $r$ such that $k^2 \leq 0.2 \,m 2^m$.  If $r < 0.1\,m$ then $k < r {m \choose 0.1\,m} < r 2^{0.47\,m}$, and $k^2 \le 0.2 \,m 2^m$ for sufficiently large $m$.  

Next, we examine the automorphism group of the Reed-Muller codes. Let $\GL_m(\FF_2)$ denote the set of invertible $m \times m$ matrices over $\FF_2$.  It is known~\citep{Ref_Sidelnikov94public,Ref_Macwilliams78theory} that $\aut(\RM(r,m))$ coincides with the general affine group of the space $\FF_2^m$. In other words, $\aut(\RM(r,m))$ consists of all affine permutations of the form $\sigma_{A,\beta}(x) = Ax+\beta$ where $A \in \GL_m(\FF_2)$ and $\beta \in \FF_2^m$.  Hence the size of $\aut(\RM(r,m))$ is
\[
|\aut(\RM(r,m))| = |\GL_m(\FF_2)|\cdot |\FF_2^m| \leq 2^{m^2+m} =2^{O(\log^2 n)} \leq \e^{o(n)}.
\]
Finally, we compute the minimal degree of $\aut(\RM(r,m))$ as follows.
\begin{proposition}
The minimal degree of $\aut(\RM(r,m))$ is exactly $2^{m-1}=n/2$.
\end{proposition}
\begin{proof}
The minimal degree is at most $2^{m-1}$, since there is an affine transformation with support $2^{m-1}$.  For example, let $A$ be the $m \times m$ binary matrix with $1$s on the diagonal and the $(1,m)$-entry and $0$s elsewhere.  Then $\sigma_{A,0}$ fixes the subspace spanned by the first $m-1$ standard basis vectors.  Its support is the complement of this subspace, which has size $2^{m-1}$.

Conversely, if $\sigma_{A,\beta}$ fixes a set $S$ that spans $\FF_2^m$, then $\sigma_{A,\beta}$ must be the identity.  To see this, let $x_0 \in S$ and consider the translated set $S'=S-x_0$.  Then $Ay=y$ for any $y \in S'$, since 
\[
y+x_0 = \sigma_{A,\beta}(y+x_0) = Ay + \sigma_{A,\beta}(x_0) = Ay + x_0 \, . 
\]
If $S$ spans $\FF_2^m$ then so does $S'$, in which case $A=\one$.  Then $\beta=0$, since otherwise $\sigma_{\one,\beta}$ doesn't fix anything, and $\sigma_{\one,0}$ is the identity.  

Moreover, any set $S$ of size greater than $2^{m-1}$ spans $\FF_2^m$.  To see this, let $B$ be a maximal subset of $S$ consisting of linearly independent vectors.  Since $B$ spans $S$, we have $|S| \le 2^{|B|}$.  Thus if $|S| > 2^{m-1}$ we have $|B| = m$, so $B$ and therefore $S$ span $\FF_2^m$.  Thus no nonidentity affine transformation can fix more than $2^{m-1}$ points, and the minimal degree is at least $2^{m-1}$.
\end{proof}

We have proved the following:
\begin{theorem}
\label{thm:reed-muller}
Reed-Muller codes $\RM(r,m)$ with $r \le 0.1\,m$ and $m$ sufficiently large are HSP-hard. 
\end{theorem}
\noindent
In the original proposal of~\citet{Ref_Sidelnikov94public}, $r$ is taken to be a small constant, where the Reed-Muller codes have low rate.  It is worth noting that the attack of~\citet{Ref_Minder07Cryptanalysis} becomes infeasible in the high-rate case where $r$ is large, due to the difficulty of finding minimum-weight codewords, while Theorem~\ref{thm:reed-muller} continues to apply.  However, as those authors point out, taking large $r$ degrades the performance of Reed-Muller codes, and presumably opens the Sidelnikov system to other classical attacks.  

\paragraph{Acknowledgments.} We thank Kirill Morozov, Nicolas Sendrier, and Daniel Bernstein for teaching us about the support splitting algorithm. We are grateful to Kirill Morozov and Nicolas Sendrier for discussions about the security of the Sidelnikov cryptosystem and methods for oblivious transfer with McEliece-type cryptosystems~\citep{DBLP:conf/mmics/KobaraMO08}. Finally, we thank the organizers of the Schloss Dagstuhl Workshop on Quantum Cryptanalysis.


\begin{thebibliography}{16}
\providecommand{\natexlab}[1]{#1}
\providecommand{\url}[1]{\texttt{#1}}
\expandafter\ifx\csname urlstyle\endcsname\relax
  \providecommand{\doi}[1]{doi: #1}\else
  \providecommand{\doi}{doi: \begingroup \urlstyle{rm}\Url}\fi

\bibitem[Babai et~al.(1982)Babai, Grigoriev, and Mount]{bounded-mult}
L.~Babai, D.~Grigoriev, and D.~Mount.
\newblock Isomorphism of directed graphs with bounded eigenvalue multiplicity.
\newblock In \emph{Proc. 14th Symposium on Theory of Computing}, pages
  310--324, 1982.

\bibitem[Babai(1980)]{babai-canonical}
L{\'a}szl{\'o} Babai.
\newblock On the complexity of canonical labeling of strongly regular graphs.
\newblock \emph{SIAM J. Computing}, 9\penalty0 (1):\penalty0 212--216, 1980.

\bibitem[Babai and Luks(1983)]{babai-luks}
L{\'a}szl{\'o} Babai and Eugene~M. Luks.
\newblock Canonical labeling of graphs.
\newblock In \emph{Proc. 15th Symposium on Theory of Computing}, pages
  171--183, 1983.

\bibitem[Dinh et~al.(2011)Dinh, Moore, and Russell]{Ref_Dinh11McEliece}
Hang Dinh, Cristopher Moore, and Alexander Russell.
\newblock {McEliece} and {Niederreiter} cryptosystems that resist quantum
  {Fourier} sampling attacks.
\newblock In \emph{Advances in Cryptology - CRYPTO 2011: 31st Annual Cryptology
  Conference Proceedings}, volume 6841 of \emph{Lecture Notes in Computer
  Science}, pages 761--779. Springer-Verlag, August 2011.

\bibitem[Faug{\'e}re et~al.(2010)Faug{\'e}re, Otmani, Perret, and
  Tillich]{FOPTDist_YACC10}
J.-C. Faug{\'e}re, A~Otmani, L.~Perret, and J.-P. Tillich.
\newblock A distinguisher for high rate mceliece cryptosystem -- extended
  abstract.
\newblock In \emph{Yet Another Conference on Cryptography, YACC 2010}, pages
  1--4, 2010.

\bibitem[Hallgren et~al.(2010)Hallgren, Moore, R\"{o}tteler, Russell, and
  Sen]{Ref_HallgrenJACM}
Sean Hallgren, Cristopher Moore, Martin R\"{o}tteler, Alexander Russell, and
  Pranab Sen.
\newblock Limitations of quantum coset states for graph isomorphism.
\newblock \emph{Journal of the ACM}, 57\penalty0 (6), 2010.

\bibitem[Kobara et~al.(2008)Kobara, Morozov, and
  Overbeck]{DBLP:conf/mmics/KobaraMO08}
Kazukuni Kobara, Kirill Morozov, and Raphael Overbeck.
\newblock Coding-based oblivious transfer.
\newblock In Jacques Calmet, Willi Geiselmann, and J{\"o}rn M{\"u}ller-Quade,
  editors, \emph{Mathematical Methods in Computer Science, MMICS 2008,
  Karlsruhe, Germany, December 17-19, 2008 - Essays in Memory of Thomas Beth},
  volume 5393 of \emph{Lecture Notes in Computer Science}, pages 142--156.
  Springer, 2008.

\bibitem[Luks(1982)]{luks}
Eugene~M. Luks.
\newblock Isomorphism of graphs of bounded valence can be tested in polynomial
  time.
\newblock \emph{J. Computer and System Sciences}, 25\penalty0 (1):\penalty0
  42--65, 1982.

\bibitem[MacWilliams and Sloane(1978)]{Ref_Macwilliams78theory}
F.J. MacWilliams and N.J.A. Sloane.
\newblock \emph{The theory of error correcting codes}.
\newblock North-Holland mathematical library. North-Holland Pub. Co., 1978.
\newblock ISBN 9780444851932.

\bibitem[Minder and Shokrollahi(2007)]{Ref_Minder07Cryptanalysis}
Lorenz Minder and Amin Shokrollahi.
\newblock Cryptanalysis of the {S}idelnikov cryptosystem.
\newblock In \emph{EUROCRYPT'07}, pages 347--360, 2007.

\bibitem[Moore et~al.(2008)Moore, Russell, and Schulman]{mrs-symmetric}
Cristopher Moore, Alexander Russell, and Leonard~J. Schulman.
\newblock The symmetric group defies strong {F}ourier sampling.
\newblock \emph{SIAM Journal on Computing}, 37:\penalty0 1842--1864, 2008.

\bibitem[Moore et~al.(2010)Moore, Russell, and {\'S}niady]{sieve}
Cristopher Moore, Alexander Russell, and Piotr {\'S}niady.
\newblock On the impossibility of a quantum sieve algorithm for graph
  isomorphism.
\newblock \emph{SIAM J. Computing}, 39\penalty0 (6):\penalty0 2377--2396, 2010.

\bibitem[Petrank and Roth(1997)]{Ref_Petrank97code}
E.~Petrank and R.M. Roth.
\newblock Is code equivalence easy to decide?
\newblock \emph{IEEE Transactions on Information Theory}, 43\penalty0
  (5):\penalty0 1602 -- 1604, 1997.
\newblock \doi{10.1109/18.623157}.

\bibitem[Sendrier(2000)]{Ref_Sendrier00finding}
Nicolas Sendrier.
\newblock Finding the permutation between equivalent linear codes: {The}
  support splitting algorithm.
\newblock \emph{IEEE Transactions on Information Theory}, 46\penalty0
  (4):\penalty0 1193 -- 1203, 2000.

\bibitem[Sidelnikov(1994)]{Ref_Sidelnikov94public}
V.~M. Sidelnikov.
\newblock A public-key cryptosystem based on binary {Reed}-{Muller} codes.
\newblock \emph{Discrete Mathematics and Applications}, 4\penalty0
  (3):\penalty0 191--207, 1994.

\bibitem[Spielman(1996)]{spielman-stronglyreg}
Daniel~A. Spielman.
\newblock Faster isomorphism testing of strongly regular graphs.
\newblock In \emph{Proc. 28th Symposium on Theory of Computing}, pages
  576--584, 1996.

\end{thebibliography}
\end{document}